\documentclass[conference]{IEEEtran}
\usepackage{amssymb}
\usepackage{amsmath}
\usepackage{amsfonts}
\usepackage{graphicx}
\usepackage{algorithm}
\usepackage{algorithmic}
\usepackage{multirow}
\usepackage{accents}
\usepackage{cite}
\usepackage{epstopdf}
\usepackage{stfloats}
\usepackage{color}

\newtheorem{theorem}{\textbf{Theorem}}

\newtheorem{definition}{\textbf{Definition}}

\newtheorem{remark}{\textbf{Remark}}

\IEEEoverridecommandlockouts

\begin{document}

\title{Signal Shaping for Two-User Gaussian Multiple Access Channels with Computation: Beyond the Cut-Set Bound}
\author{Zhiyong Chen, and Hui Liu\\
Cooperative Medianet Innovation Center, Shanghai Jiao Tong University, Shanghai, P. R. China\\
Email: {\{zhiyongchen, huiliu\}@sjtu.edu.cn}
}
\maketitle

\begin{abstract}
In this paper, we investigate the signal shaping in a two-user discrete time memoryless Gaussian multiple-access channel (MAC) with computation. It is shown that by optimizing input probability distribution, the transmission rate per transmitter is beyond the cut-set bound. In contrast with the single-user discrete memoryless channel, the Maxwell-Boltzmann distribution is no longer a good approximation to the optimal input probability distribution for this discrete-time Gaussian MAC with computation. Specifically, we derive and analyze the mutual information for this channel. Because of the computation in the destination, the mutual information is not concave in general on the input probability distribution, and then  primal-dual interior-point method is used to solve this non-convex problem. Finally, some good input probability distributions for 16-ary pulse amplitude modulation (PAM) constellation are obtained and achieve $4.0349$ dB gain over the cut-set bound for the target transmission rate $3.0067$ bits/(channel use).
\end{abstract}

\section{Introduction, Signal Model and Main Results}
In this paper, we study a two-user discrete-time memoryless Gaussian multiple access channel (MAC) with computation. The two-user discrete-time memoryless Gaussian MAC is defined by two input alphabets $\mathcal{X}_{i}$, $i=A,B$, and output alphabet $\mathcal{Y}$, and a conditional probability distribution $\Pr(Y|X_{A},X_{B})$. Here, let $X_{i}$ and $Y$ be random variables taking values in $\mathcal{X}_{i}$ and $\mathcal{Y}$, respectively. We consider $Y=X_{A}+X_{B}+Z$ and $Z\thicksim \mathcal{CN}(0,N_{0})$ is a Gaussian noise as shown Fig. \ref{systemodel2}. $Z$ is independent of $X_{i}$, $i=A,B$. Furthermore, we consider $\mathcal{X}_{A}$ and $\mathcal{X}_{B}$ are one-dimensional signal constellations, but $\mathcal{X}_{A}$ is orthogonal with $\mathcal{X}_{B}$. For example, we use $\mathcal{X}_{A}=\{a_{i}\}_{i=1}^{M}$ and $\mathcal{X}_{B}=\{\sqrt{-1}a_{i}\}_{i=1}^{M}=\sqrt{-1}\mathcal{X}_{A}$ with corresponding probability $P_{X}=\{p_{i}\}_{i=1}^{M}$, where $a_{i}$ is real number, $i=1,...,N$.

Now, we describe the \emph{computation} operation. Source $i$ has source bit messages $W_{i}$, $i=A,B$. $W_{A}$ and $W_{B}$ are independent. By using linear modulation, $W_{i}$ is mapped to the signal symbol $X_{i}$, $i=A,B$. Thus, $X_{A}$ is also independent of $X_{B}$. Different from the conventional MAC, the goal of the destination in this paper is to compute a target mod-2 sum of the messages from the received signals $Y$, i.e., $W_{C}=W_{A}\oplus W_{B}$. In this context, different $x_{A}+x_{B}$ values can represent the same $w_{C}$ value, due to the many-to-one operation of computation in the destination.

Consider the probability restriction $\sum_{i=1}^{M}p_{i}=1$ and $0\leq p_{i}\leq 1$ for $\forall i$, and the average transmission power constraint $\mathbb{E}[|x_{A}|^2]=\mathbb{E}[|x_{B}|^2]=\sum_{i=1}^{M}p_{i}|a_{i}|^2\leq P$. This model can be regarded as the multiple access phase in two-way relaying channels \cite{CutSetBound,UBound}. To the best of my knowledge, the capacity of such channels is still unknown, but it is upper bounded by
     \begin{equation}
 C_{cs}=\frac{1}{2}\log_{2}(1+\frac{P}{\sigma^{2}}),\label{cutset}
    \end{equation}
per transmitter based on the \textbf{cut-set bound} \cite{CutSetBound,UBound}. Here, we use $\sigma^{2}=N_{0}/2$ as the noise variance
per dimension. As well known, this upper bound is attained by continuous Gaussian input. For this discrete distribution inputs $\mathcal{X}_{A}$ and $\mathcal{X}_{B}$, \emph{what is the optimal signalling strategy $P_{X}$ in order to maximize the mutual information $I(W_{C};Y)$ in this two-user discrete time memoryless Gaussian MAC with computation $\Pr(y|x_{A},x_{B})$?}

\subsection{Main Results}
In this paper, we answer this fundamental problem in some cases. In particular, we consider a real-value $2^{m}$-ary pulse amplitude modulation (PAM) constellation for $\mathcal{X}_{A}$ and a complex-value $2^{m}$-ary PAM constellation for $\mathcal{X}_{B}$, i.e., $M=2^{m}$ and $a_{i}=-M+2i-1$ for $i=1,...,M$. With this orthogonal constellations, the destination can do the ambiguity-free detection, which means $w_{C}$ can be uniquely decoded.

\begin{figure}
   \centering
    \includegraphics[width=2.7in]{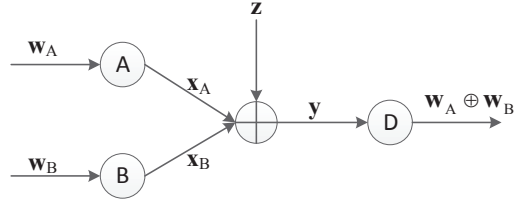}
  \centering
  \caption{A two-user discrete-time memoryless Gaussian multiple access channel with computation.}\label{systemodel2}
\end{figure}
\subsubsection{Good Input Probability Distributions}
We first derive the mutual information $I(W_{C};Y)$ for the transmission of arbitrary $\mathcal{X}_{A}$ and $\mathcal{X}_{B}$ with $P_{X}$. The optimal signalling strategy problem for an optimal choice $P_{X}^{*}$ to maximize $I(W_{C};Y)$ is then formulated. In contrast with the single-user discrete memoryless channel (DMC) where the problem of capacity computation is convex, this convexity is missing. As a result, we use primal dual interior-point method to carry out this optimization problem and obtain good input probability distributions.

The \emph{Maxwell-Boltzmann} (MB) distribution and the uniform distribution are considered as two benchmarks in this paper. The Maxwell-Boltzmann distribution provides a very good approximation to the optimal distribution obtained from the Blahut-Arimoto algorithm for the single-user DMC \cite{BA-B,BA-A,MB}, and can be written as
  \begin{equation}\label{MBdistribution}
  p_{i}=\frac{\exp\left(-\lambda|a_{i}|^2\right)}{\sum_{i}\exp\left(-\lambda|a_{i}|^2\right)},
  \end{equation}
where the parameter $\lambda$ characterizes the trade-off between the average power $P$ and entropy $H(W_{c})$. We can consume the minimum average energy which means the minimum signal-to-noise (SNR) to achieve a given transmission rate $R_{t}$ by selecting $\lambda$ properly. Taking $R_t=3.0067$ (bits per channel use) for example, the optimized value $\lambda^*=0.0295$ for 16-PAM. We have $p_{i}=1/M$ for the uniform distribution.

\begin{table}[t]
\renewcommand{\arraystretch}{1}
\centering
\caption{Good input probability distributions $P_{X}^{*}$ of 16-PAM constellation for Gaussian MAC over computation with constraints on the different target transmission rate $R_{t}$ (bits/channel use). For each input distribution the SNR threshold $(\frac{P}{\sigma^{2}})^*$ is given. Besides, applying the cut-set bound (\ref{cutset}) for the target rate $R_{t}$, we can get $(\frac{P}{\sigma^{2}})^{cs}$. We also apply the MB distribution and uniform distribution to the constellation, and obtain $(\frac{P}{\sigma^{2}})^{MB}$ and $(\frac{P}{\sigma^{2}})^{uf}$ for $R_{t}$, respectively.}\label{shapingencoder}
\begin{tabular}{c|c|c|c|c }
  \hline
 $R_t$& 3.0067 & 1.9724 & 0.9846 & 0.5239 \\
  \hline
  \hline
  $p_{1}$& $0.0002276$ & 0.0000896 & 0.0000140 &0.0184170\\
    \hline
  $p_{2}$&0.0933530 &0.0000637 &  0.0000495&0.0540220\\
    \hline
      $p_{3}$&0.0001605 &  0.0001011& 0.0250460& 0.0003871\\
    \hline
      $p_{4}$&0.0001664 &0.0845740& 0.0464800& 0.0560510\\
    \hline
      $p_{5}$& 0.1692600&0.1336100&0.0062702&0.0000330\\
    \hline
      $p_{6}$& 0.1833400&0.0001079& 0.0229170&0.0000371\\
    \hline
  $p_{7}$&0.0002105 &0.2013800&0.0000193&0.0000140\\
    \hline
      $p_{8}$&  0.2352551& 0.3054700&0.7495156&0.0000319\\
    \hline
      $p_{9}$& 0.0001460 & 0.0000463&0.0000159&0.7945245\\
    \hline
      $p_{10}$&0.0001526 &0.0000470&0.0000211&0.0000120\\
    \hline
      $p_{11}$&0.2007700 &0.1881500&0.0001262&0.0000012\\
    \hline
  $p_{12}$& 0.0002296&0.0001013& 0.0903270&0.0000113\\
    \hline
      $p_{13}$&0.0002384&0.0860131 & 0.0589980&0.0000281\\
    \hline
      $p_{14}$& 0.1157200&0.0000909& 0.0001697&0.0004899\\
    \hline
      $p_{15}$&  0.0001761&0.0000615&0.0000167&0.0256960\\
    \hline
          $p_{16}$&  0.0005942& 0.0000936&  0.0000138&0.0502440\\
    \hline
    \hline
    %$\sigma^*$ &  0.1995262&  0.3548134 &  0.7079458&1.1220184\\
   % \hline
    $(\frac{P}{\sigma^{2}})^{*}$ dB &  14 &9 &3&-1\\
    \hline
    $(\frac{P}{\sigma^{2}})^{cs}$ dB & 18.0349 &  11.5834 &4.6471 & 0.2831\\
    \hline
        $(\frac{P}{\sigma^{2}})^{MB}$ dB & 17.7858 &  11.5893&5.5165&2.3252\\
    \hline
        $(\frac{P}{\sigma^{2}})^{uf}$ dB & 22.0250 &  16.9335 &10.1647&5.8724\\
    \hline
\end{tabular}
\end{table}

The resulting input probability distributions of 16-PAM\footnote{It is easily to obtain good input probability distributions with any $n$ or other constellations, e.g., quadrature amplitude modulation (QAM).} ($M=16$) are shown in Table I for $R_{t}=3.0067, 1.9724, 0.9846$ and $0.5239$ (bits per channel use). Each column corresponds to on particular input probability distribution $\varphi_{p}^{*}$ as well as the corresponding SNR $((\frac{P}{\sigma^2})^*$ (dB)). Here, the threshold $(\frac{P}{\sigma^2})^*$ is the smallest required $\frac{P}{\sigma^2}$ such that the transmission achieves a given rate $R_{t}$. Likewise, we have the thresholds $(\frac{P}{\sigma^2})^{MB}$ and $(\frac{P}{\sigma^2})^{uf}$ for the MB distribution and the uniform distribution, respectively. Interestingly, we can observe from Table I that:
\begin{itemize}
  \item  \textbf{The uniform distribution suffers a large shaping loss}. For example, when the achievable rate $R_{t}$ is $3.0067$ bits/(channel use), the uniform distribution is far away from the cut-set bound by $3.9901$ dB, which points out the importance of signal shaping.
  \item  \textbf{In contrast with the single-user DMC, the Maxwell-Boltzmann distribution is no longer a good approximation to the optimal distribution for discrete-time Gaussian MAC with computation.} Similarly, considering $R_{t}=1.9724$ bits/(channel use), the proposed good input probability distribution $P_{X}^{*}$ achieves $2.5893$ dB gain compared with the optimized MB distribution with $\lambda^*=0.115$.
\end{itemize}
\subsubsection{Beyond the cut-set bound}
Surprisingly, we find that in discrete-time Gaussian MAC with computation, \textbf{the transmission rate with the proposed good input probability distribution can \emph{beat} the cut-set bound!} Based on the cut-set bound (\ref{cutset}), the threshold $(\frac{P}{\sigma^2})^{cs}$ on this channel is $18.0349$ dB for $C_{cs}=3.0067$ bits/(channel use), but we can get a much better threshold $(\frac{P}{\sigma^2})^{*}=14$ dB based on the proposed good input probability distribution $P_{X}^{*}$. Similarly, the proposed good input distribution can outperform the cut-set bound by $2.5834$, $1.6471$ and $1.2831$ dB for $R_t=1.9724, 0.9846$ and $0.5239$ bits/(channel use), respectively.

Because $I(W_{C};Y)$ for this channel cannot be calculated in closed form, we cannot prove this phenomenon theoretically. However, we can express this observation from the \emph{mismatch} of the cut-set bound for discrete-time Gaussian MAC with computation. Accordingly, the cut-set bound is defined as $C_{cs}=\max_{P_{X}} I(X_{A};Y|X_{B})$, which means the maximum rate achievable from source A to the destination when source B is not sending any information. In this context, when the transmission rate is below $C_{cs}$, $X_{A}$ can be reliably transmitted with arbitrarily small error probability $\Pr(w_{A}\neq \hat{w}_{A})$. However, the destination wants to obtain $w_{A}\oplus w_{B}$, rather than individual $w_{A}$ or $w_{B}$. Due to the computation, more than one superposition signals $x_{A}+x_{B}$ are mapped to one $w_{A}\oplus w_{B}$, yielding $w_{A}\oplus w_{B}=w^{'}_{A}\oplus w^{'}_{B}$ for $w_{A}\neq w^{'}_{A}$ and $w_{B}\neq w^{'}_{B}$. Thus, there maybe exist some cases such that the transmission rate per transmitter maybe can exceed $C_{cs}$, yielding $\Pr(w_{C}\neq \hat{w}_{C})\rightarrow 0$, although we have $\Pr(w_{A}\neq \hat{w}_{A})\neq 0$ or $\Pr(w_{B}\neq \hat{w}_{B})\neq 0$.

Besides, this cut-set bound $C_{cs}$ can be attained by Gaussian input $x_{A}\sim N(0, P)$. In general, there is no solution to calculate $I(W_{c};Y)$ with $X_{A}\sim N(0, P)$ and $X_{B}\sim N(0, P)$. But according to (\ref{MBdistribution}), the MB distribution is a discrete Gaussian distribution. We can use the MB distribution to verify the above-mentioned phenomenon. From Table I, we can see that the transmission rate of the MB distribution can also beat the cut-set bound at some SNR values. For example, the MB distribution with $\lambda^*=0.0295$ achieves $0.2491$ dB gain over the cut-set bound for rate $R_{t}=3.0067$ bits/(channel uses),

Accordingly, we \emph{conjecture} that \textbf{the cut-set bound is not tight for the discrete-time Gaussian MAC with computation, and can be exceeded.}

\setcounter{equation}{6}
\begin{figure*}[b]
\hrulefill
\begin{align}\label{CCCapacity}
&I(W_{C};Y)=\sum_{y}\sum_{i=1}^{M}\left(\sum_{(k,l)\in\Omega_{i}}p_{k}p_{l}\right)\Pr(y|w_{C}^{i})\log\frac{\Pr(y|w_{C}^{i})}{\sum_{j=1}^{M}\left(\sum_{(k^{'},l^{'})\in\Omega_{j}}p_{k'}p_{l'}\right)\Pr(y|w_{C}^{j})}\\
%&=\sum_{y}\sum_{i=1}^{n}\left(\sum_{(k,l)\in\Omega_{i}}p_{k}p_{n}\Pr(y|x_{AB}(k,l))\right)\log\frac{\sum_{(k,l)\in\Omega_{i}}p_{k}p_{n}\Pr(y|x_{AB}(k,l))}{\left(\sum_{(k,l)\in\Omega_{i}}p_{k}p_{l}\right)\sum_{j=1}^{n}\left(\sum_{(k^{'},l^{'})\in\Omega_{j}}p_{k^{'}}p_{l^{'}}\Pr(y|x_{AB}(k^{'},l^{'}))\right)}\nonumber\\
&=-\sum_{i=1}^{M}\left(\sum_{(k,l)\in\Omega_{i}}p_{k}p_{l}\log\sum_{(k,l)\in\Omega_{i}}p_{k}p_{l}\right)+\sum_{y}\sum_{i=1}^{M}\left(\sum_{(k,l)\in\Omega_{i}}p_{k}p_{l}\Pr(y|x_{AB}^{i}(k,l))\right)\log\frac{\sum_{(k,l)\in\Omega_{i}}p_{k}p_{l}\Pr(y|x_{AB}^{i}(k,l))}{\sum_{j=1}^{M}\left(\sum_{(k^{'},l^{'})\in\Omega_{j}}p_{k^{'}}p_{l^{'}}\Pr(y|x_{AB}^{j}(k^{'},l^{'}))\right)}.\nonumber
\end{align}
\end{figure*}
\subsection{Related Work}
For the single-user DMC, the mutual information is a concave function of the input probability distribution and Kuhn-Tucker condition is necessary and sufficient for a distribution to maximize the mutual information. The Blahut-Arimoto algorithm is then developed to compute the optimal input probability distribution \cite{BA-B,BA-A}, and can be approached by the MB distribution \cite{MB}. Indeed, by selecting constellation points properly based on the MB distribution at any dimension, the ultimate shaping gain (1.53 dB) can be achieved \cite{MB}.

For discrete-time memoryless Gaussian MAC, the mutual information (e.g., $I(X_{A},X_{B};Y)$) is not concave on the input probability distribution. However, with the binary input, the total capacity can be calculated for a two-user discrete-time memoryless Gaussian MAC \cite{MAC-Bin}. Then, a generalized Blahut-Arimoto algorithm has been developed for computation of the total capacity of discrete-time memoryless Gaussian MAC \cite{MAC-BA}. More recently, a two-user Gaussian MAC under peak power constraints at the transmitters is addressed in \cite{MAC-bound}, which proves that discrete distributions with a finite number of mass points can achieve any point on the boundary of the capacity region.

Instead of reconstructing all the signals of each transmitter, the destination only reconstructs a function of sources in a MAC over computation \cite{MAC-Com}. The work of \cite{MAC-Com} presents that structured codes can achieve higher computation rates for computing the modulo-sum of two messages in Gaussian MAC over computation. Accordingly, \cite{CutSetBound} achieves a rate of $\frac{1}{2}\log(\frac{1}{2}+\frac{P}{\sigma^2})$ by using lattice coding in the multiple access phase of the two-way relaying channels. The results are effective in understanding specific features of different computation functions inherent the model. However, there remain much fundamental problems to be done, e.g., the optimal input probability distribution for discrete memoryless Gaussian MAC with computation.
\section{Problem formulation and Analysis}
\begin{definition}
Let $x_{AB}=x_{A}+ x_{B}$ denote the superimposed signal without noise in the destination.  Let $\mathcal{V}_{i}=\{x_{AB}^{i}|w_{C}^{i}\}$ denote the signal set with respect to a given computation messages $w_{C}^{i}$, for $i=1,\cdots,M$. We use natural mapping for $2^m$-ary PAM in this paper, e.g., $w_{C}^{3}=0010$ for $M=16$. There are $M$ different pairs of $(w_{A}, w_{B})$ associated with the same $w_{C}^{i}$. Thus, the cardinality of the signal set is $|\mathcal{V}_{i}|=M$.
 \end{definition}
 \begin{definition}
Let $\Omega_{i}=\{(k,l):x_{AB}^{i}(k,l)=a_{k}+\sqrt{-1}a_{l}\in \mathcal{V}_{i},1\leq k\leq M, 1\leq l \leq M\}$ denote the index pairs set of $(x_{A}, x_{B})$ of associated with the same $w_{C}^{i}$. We then have $|\Omega_{i}|=M$ and the probability of $w_{C}^{i}$ can be calculated as
\setcounter{equation}{2}
\begin{equation}
\Pr(w_{C}^{i})= \sum_{(k,l)\in \Omega_{i}}p_{k}p_{l},\label{pr_W}
\end{equation}
for $i=1,\cdots, M$.
 \end{definition}
\begin{remark}
Therefore, the entropy of $W_{C}$ is
\begin{equation}
H(W_{C})=-\sum_{i=1}^{M}\sum_{(k,l)\in\Omega_{i}}p_{k}p_{l}\log_{2}{\sum_{{(k,l)}\in\Omega_{i}}p_{k}p_{l}}.
\end{equation}

The first- and second-order derivatives of $H(W_{C})$ are
\begin{align}
\frac{\partial H(W_{C})}{\partial p_{i}}=-\frac{2}{\ln2}-2\sum_{j=1}^{M}p_{j}^{i}\left(\log_2\sum_{(k,l)\in\Omega_{j}}p_{k}p_{l}\right),
\end{align}

\begin{align}
\frac{\partial^{2} H(W_{C})}{\partial p_{i}\partial p_{n}}&=-2\log_{2}\sum_{(k,l)\in\Omega_{j'}}p_{k}p_{l}-\frac{4}{\ln2}\frac{p_{i}p_{n}}{\sum_{(k,l)\in\Omega_{j'}}p_{k}p_{l}},\nonumber\\
&-\sum_{j=1,j\neq j'}^{M}\frac{4p_{j}^{i}p_{j}^{n}}{\ln2\sum_{(k,l)\in\Omega_{j}}p_{k}p_{l}}, \label{EntSe}
\end{align}
where $p_{j}^{i}$ is the probability of $a_{k}$ or $\sqrt{-1}a_{k}$, yielding $a_{k}+\sqrt{-1}a_{i}\in\Omega_{j}$ or $a_{i}+\sqrt{-1}a_{k}\in\Omega_{j}$ for given $i$, respectively. $\Omega_{j'}$ denotes the index pairs set of $(x_{A}, x_{B})$ including of $(k,l)=(i,n)$ or $(k,l)=(n,i)$. We can see from (\ref{EntSe}) that $\frac{\partial^{2} H(W_{C})}{\partial p_{i}\partial p_{n}}$ is dependent of $P_{X}$, so that the convexity of $H(W_{C})$ on $P_{X}$ is missing.
\end{remark}

Fortunately, because $x_{A}$ is orthogonal with $x_{B}$, we have $\Pr(y|x_{A},x_{B})=\Pr(y|x_{AB})$ for all $(x_{A},x_{B})\in\mathcal{X}_{A} \times \mathcal{X}_{B}$. The mutual information $I(W_{C};Y)$ is then given by the following theorem.
\begin{theorem}
The mutual information $I(W_{C};Y)$ of a discrete-time memoryless Gaussian MAC with computation $\Pr(y|x_{A},x_{B})$ is shown in (\ref{CCCapacity}).
\end{theorem}
\begin{proof}
Given $w_{C}^i$, $i=1,\cdots,M$, the conditional probability $\Pr(y|w_{C}^{i})$ can be written as
\setcounter{equation}{7}
\begin{align}\label{prob_2}
\Pr(y|w_{C}^{i})&=\frac{\Pr(w_{C}^{i}|y)\Pr(y)}{\Pr(w_{C}^{i})}=\sum\limits_{(k,l)\in\Omega_{i}}{\frac{{\Pr(x_{AB}^{i}(k,l)|y)}\Pr(y)}{\Pr(w_{C}^{i})}}\nonumber\\
&=\sum\limits_{(k,l)\in\Omega_{i}}{\frac{{\Pr(y|x_{AB}^{i}(k,l))}\Pr(x_{AB}^{i}(k,l))}{\Pr(w_{C}^{i})}}\nonumber\\
&\overset{(i)}{=}\sum\limits_{(k,l)\in\Omega_{i}}{\frac{{p_{k}p_{l}}\Pr(y|x_{AB}^{i}(k,l))}{\sum_{(k,l)\in \Omega_{i}}p_{k}p_{l}}}
\end{align}
where Steps $(i)$ is based on $\Pr(x_{AB}^{i}(k,l))=p_{k}p_{l}$. Accordingly, we have Theorem 1.
\end{proof}

According to (\ref{prob_2}), $\Pr(y|w_{C}^{i})$ is dependent on $P_{X}$ given the channel transform matrix $\Pr(y|x_{A},x_{B})$. If $P_{X}$ is uniform distribution, i.e., $\Pr_{i}=\frac{1}{M}$ for all $i$, $\Pr(y|w_{C}^{i})$ becomes $\frac{1}{M}\sum_{(k,l)\in\Omega_{i}}{\Pr(y|x_{AB}^{i}(k,l))}$.

\begin{remark}
The mutual information $I(W_{C};Y)$ can also be written as $I(W_{C};Y)=\sum_{y}\sum_{i=1}^{M}\Pr(w_{C}^{i})\Pr(y|w_{C}^{i})\log\frac{\Pr(y|w_{C}^{i})}{\sum_{j=1}^{M}\Pr(w_{C}^{i})\Pr(y|w_{C}^{i})}$. According to the concavity of mutual information, for fixed $\Pr(y|w_{C}^{i})$, $i=1,...,M$, $I(W_{C};Y)$ is a concave functional of $\Pr(w_{C})$. However, the channel transform matrix is $\Pr(y|x_{A},x_{B})$, not $\Pr(y|w_{C})$.
\end{remark}

Therefore, the optimal input probability distribution problem for a discrete-time memoryless Gaussian MAC with computation $\Pr(y|x_{A},x_{B})$  can be formulated as
\begin{align}
&\max_{P_{X}}\quad I(W_{C};Y) \nonumber\\
&\quad\mathnormal{s.t.}\quad\sum_{i=1}^{M}p_{i}=1 \nonumber\\
&~\quad\quad\quad p_{i}\geq0, \quad i=1,...M \nonumber\\
&\quad\quad\quad\sum_{i=1}^{M}p_{i}|a_{i}|^2\leq P \label{condition3}
\end{align}

Based on $\sum_{i=1}^{M}p_{i}=1$, $P_{X}$ is located on an $(M-1)$-dimensional simplex $D_{p}$. Similar to \cite{MAC-Bin}, $I(W_{C};Y)$ is not in general concave on the input probability distribution $P_{X}$. Because $\frac{\partial \Pr_(y|w_{C}^{j})}{\partial p_{i}}$ does not equal $0$ and $I(W_{C};Y)$ is a function of $P_{X}\bigotimes P_{X}$, it is difficult to evaluate obtain the necessary condition for the optimal problem based on Kuhn-Tucker condition. Here, $\bigotimes$ denotes Kronecker product.

For this non-convex problem, we can use some stochastic optimization algorithm to solve this non-convex problem \cite{convex}, e.g, genetic algorithm or annealing algorithm. However, stochastic algorithms have high complexity and cannot guarantee a global optimal solution. In this paper, we use primal-dual interior-point algorithm with random initial values to find the optimal solution \cite{trust_region}. Primal-dual interior-point method is a deterministic optimization algorithm, where every feasible initial values is related to a local optimal solution. As a result, we can use different initial values to approach the global optimal solution.

%\section{Analysis and Algorithm}
%As in \cite{MAC-Bin}, we have the following proposition to state the necessary condition for optimal input distribution $\varphi^{*}_{p}$.
%\begin{proposition}
%We
%\end{proposition}
%\begin{proof}
%Consider the Lagrangian
%\begin{equation}
%L(\varphi_{x})=I(w_{C};y)+\lambda(\sum_{i=1}^{n}p_{i}-1)+\sum_{i=1}^{n}\mu_{i}p_{i}+\nu(P-\sum_{i=1}^{n}p_{i}|a_{i}|^2)
%\end{equation}
%\end{proof}
\section{Numerical Results}
We consider 16-PAM constellation in the simulation. Based on Theorem 1 and primal-dual interior-point algorithm, we can search good input probability distribution $P_{X}^{*}$ and then obtain the transmission rate at each $\frac{P}{\sigma^2}$ per transmitter. Fig. \ref{capacity_16PAM} plots the transmission rate per transmitter for increasing $\frac{P}{\sigma^2}$ with different $P_{X}$, where some corresponding proposed good input probability distributions $P_{X}^{*}$ are listed in Table I.

From the simulation results, it is clearly show that with the uniform distribution, the system suffers a large shaping loss, where the shaping loss is larger than $1.53$ dB. For example, the gap at 2.5 (bits/channel use) between the cut-set bound and the achievable rate with uniform distribution is $4.67$ dB. Thus, it is very necessary to do the signal shaping. Moreover, the results show that the proposed good input probability distribution has significantly shaping gain compared with the uniform distribution, e.g., $8.38$ dB in $2.5$ bits/(channel use).
\begin{figure}[t]
   \centering
    \includegraphics[width=3.5in,height=3in]{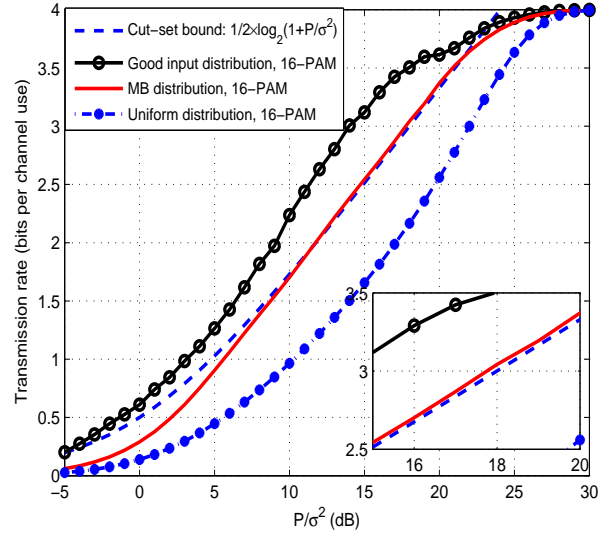}
    %\hspace{2in}\parbox{1\linewidth}
  \caption{Transmission rate per transmitter of 16-PAM with different distributions in discrete time Gaussian multiple access channel with
computation.}\label{capacity_16PAM}
\end{figure}

Interestingly, it can also be seen that for fixed $\frac{P}{\sigma^2}$ $(-5~dB\leq\frac{P}{\sigma^2}\leq25$ dB), the achievable rate per transmitter with $P_{X}^{*}$ based 16-PAM is larger than the cut-set bound. For example, for rate-2.5 bit per channel use, $P_{X}^{*}$ provides gain of 3.71 dB compared the cut-set bound. Notice also that the transmission rate based the MB distribution is very close to the cut-set bound, even more than the cut-set bound for high SNR ($11.5$ dB $\leq\frac{P}{\sigma^2}\leq21.5$ dB). The reason for this observation is presented in Section I-A.

Fig. \ref{dist_plot} plots the good distribution and the corresponding probability distribution of $W_{C}$ for $R_{t}=3.0067$ bits/(channel use). It is clear from this figure that this $P_{X}$ has more volatility than $\Pr(W_{C})$. This behavior implies the computation operation can smooth the peak-to-average probability. This can be very useful to improve the entropy. With this probability distribution, we have $H(X_{A})=H(X_{B})=2.5455$ bits, but the entropy of $W_{C}$ is increased to $H_{W_{C}}= 3.7959$ bits.

\begin{figure}[t]
   \centering
    \includegraphics[width=3.5in,height=2.7in]{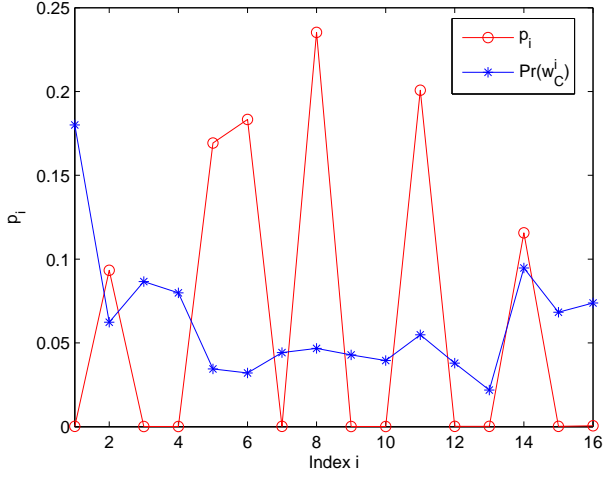}
    %\hspace{2in}\parbox{1\linewidth}
  \caption{An good input probability distribution $P_{X}=\{p_{i}\}_{i=1}^{16}$ for 16-PAM and the probability of $W_{C}$ at $R_{t}=3.0067$ bits/(channel use).}\label{dist_plot}
\end{figure}

\section{Extension and Discussion}
\subsection{Different input probability distributions for different transmitters}
In previous sections, two orthogonal 1-D constellations with the same input probability distribution $P_{X}$ are used. Obviously, it can be extended to two orthogonal 1-D constellations with different input probability distribution $P_{X}$ and $Q_{X}$. Let $Q_{X}=\{q_{i}\}_{i=1}^{M}$ be the input probability distribution of $\mathcal{X}_{B}$. Therefore, (\ref{pr_W}) becomes $\sum_{(k,l)\in \Omega_{i}}p_{k}q_{l}$ and then we can get $I(W_{C};Y)$ by substituting $\Pr(w_{C}^{i})$ into (\ref{CCCapacity}). $Q_{X}$ is also located on an $(M-1)$-dimensional simplex $Q_{p}$, and $D_{p}\bigotimes Q_{p}$ is a domain of $P_{X}\bigotimes Q_{X}$. As a result, we need to search $P_{X}^{*}$ and $Q_{X}^{*}$ in $D_{p}\bigotimes Q_{p}$ to maximize $I(W_{C};Y)$. By increasing the search dimension, $I(W_{C};Y)$ based on $P_{X}^{*}$ and $Q_{X}^{*}$ is no less than that based on $P_{X}^{*}$ and $P_{X}^{*}$.

Taking 16-PAM constellation for example, for given $\frac{P}{\sigma^2}=9$ dB, $I(W_{C};Y)$ can achieve 2.9762 bits/(channel use) based on the following $P_{X}^{*}$ and $Q_{X}^{*}$.
\begin{align}
P_{X}^{*}&=10^{-4}\times\left\{2,1226 ,   2,   2,   3,    3326,    2,    2,   2,  3701,  2,  2,  2,\right. \nonumber\\
&~~~~~~~~~~~~~\left. 1721,  2, 3\right\};\\
Q_{X}^{*}&=10^{-4}\times\left\{12,2,2,1466,2,2,3523,2,2,3517,2,2,\right. \nonumber\\
&~~~~~~~~~~~~~\left. 1457,2,2,5\right\}.
\end{align}
We can see that $P_{X}^{*}$ and $Q_{X}^{*}$ have different structure. Compared with the result in Table I, this transmission provides significant gain ($1$ bits/(channel use)) compared with the transmission based on $P_{X}^{*}$ and $P_{X}^{*}$.
\subsection{2-D Constellation}
We also can extend our results to 2-D constellation. Notice that due to the requirement of ambiguity-free detection in destination, not all 2-D constellations can be used. In addition, by using 2-D constellation, one signal $x_{AB}$ can be mapped to the same $w_{C}$ value more than one time, because different signal pairs $(x_{A},x_{B})$ may share the same signal $x_{AB}$ value \cite{MLC}, i.e. $x_{A}+x_{B}=x_{A}^{'}+x_{B}^{'}$ with $(x_{A}\neq x_{A}^{'},x_{A}\neq x_{B}^{'})$. Combining with Theorem 1 in \cite{MLC}, we also can obtain $I(W_{C};Y)$ with $P_{X}$.

Consider 16 quadrature amplitude modulation (QAM) with Gray mapping, where it is ambiguity-free. Some good input probability distributions $P_{X}^{*}$ for 16-QAM are listed in Table II at different $R_{t}$. We can see that for 16-QAM, the proposed $P_{X}^{*}$ also outperform the cut-set bound by $0.3$ and $0.57$ dB at $R_{t}=3.2494$ and $R_{t}=2.7475$ bits/(channel use), respectively.
\section{Conclusions}
We address the optimization problem of the input probability distribution to maximize the mutual information for a two-user Gaussian MAC with computation. We formulate and analyze the optimization problem, and then use the primal-dual interior-point algorithm to search the optimal input probability distribution. The main results are summarized as follows:
\begin{itemize}
  \item The uniform distribution suffers a large shaping loss compared with the cut-set bound, where the shaping loss is larger than $1.53$ dB.
  \item The Maxwell-Boltzmann distribution also suffers a large performance loss compared with the optimal input probability distribution.
  \item The proposed input probability distribution can achieve a significant gain compared with the cut-set bound.
\end{itemize}

\begin{table}[t]
\renewcommand{\arraystretch}{1}
\centering
\caption{Good input probability distribution $P_{X}^{*}$ of 16-QAM constellation with Gray mapping for Gaussian MAC over computation for different target transmission rates $R_{t}$ (bits/channel use).}\label{shapingencoder}
\begin{tabular}{c|c|c|c|c }
  \hline
 $R_t$& 3.9176 & 3.2494 & 2.7475 & 1.5412\\
  \hline
    \hline
    $(\frac{P}{2\sigma^2})^{*}$ dB &  14 &9 &7&3\\
    \hline
    $(\frac{P}{2\sigma^2})^{cs}$ dB & 11.4958 &  9.2991 & 7.5706 & 2.8112\\
    \hline
        $(\frac{P}{2\sigma^2})^{MB}$ dB & 14.025 &  9.012&7.011&3.0002\\
    \hline
        $(\frac{P}{2\sigma^2})^{uf}$ dB & 15.0681 &  10.9920 &9.1750&5.0278\\
    \hline
\end{tabular}
\end{table}

\bibliographystyle{IEEEtran}
\bibliography{IEEEabrv,BCCM_Journal}

% Generated by IEEEtran.bst, version: 1.13 (2008/09/30)
\begin{thebibliography}{10}
\providecommand{\url}[1]{#1}
\csname url@samestyle\endcsname
\providecommand{\newblock}{\relax}
\providecommand{\bibinfo}[2]{#2}
\providecommand{\BIBentrySTDinterwordspacing}{\spaceskip=0pt\relax}
\providecommand{\BIBentryALTinterwordstretchfactor}{4}
\providecommand{\BIBentryALTinterwordspacing}{\spaceskip=\fontdimen2\font plus
\BIBentryALTinterwordstretchfactor\fontdimen3\font minus
  \fontdimen4\font\relax}
\providecommand{\BIBforeignlanguage}[2]{{%
\expandafter\ifx\csname l@#1\endcsname\relax
\typeout{** WARNING: IEEEtran.bst: No hyphenation pattern has been}%
\typeout{** loaded for the language `#1'. Using the pattern for}%
\typeout{** the default language instead.}%
\else
\language=\csname l@#1\endcsname
\fi
#2}}
\providecommand{\BIBdecl}{\relax}
\BIBdecl

\bibitem{CutSetBound}
M.~P. Wilson, K.~Narayanan, H.~D. Pfister, and A.~Sprintson, ``Joint physical
  layer coding and network coding for bidirectional relaying,'' \emph{IEEE
  Trans. Inf. Theory}, vol.~56, no.~11, pp. 5641--5654, Nov 2010.

\bibitem{UBound}
W.~Nam, S.~Y. Chung, and Y.~H. Lee, ``{Capacity of the Gaussian Two-Way Relay
  Channel to Within $\frac{1}{2}$ Bit},'' \emph{IEEE Trans. Inf. Theory},
  vol.~56, no.~11, pp. 5488--5494, Dec 2010.

\bibitem{BA-B}
R.~Blahut, ``Computation of channel capacity and rate-distortion functions,''
  \emph{IEEE Transactions on Information Theory}, vol.~18, no.~4, pp. 460--473,
  Jul 1972.

\bibitem{BA-A}
S.~Arimoto, ``An algorithm for computing the capacity of arbitrary discrete
  memoryless channels,'' \emph{IEEE Trans. Inf. Theory}, vol.~18, no.~1, pp.
  14--20, Jan 1972.

\bibitem{MB}
F.~R. Kschischang and S.~Pasupathy, ``Optimal nonuniform signaling for gaussian
  channels,'' \emph{IEEE Trans. Inf. Theory}, vol.~39, no.~3, pp. 913--929, May
  1993.

\bibitem{MAC-Bin}
Y.~Watanabe, ``The total capacity of two-user multiple-access channel with
  binary output,'' \emph{IEEE Trans. Inf. Theory}, vol.~42, no.~5, pp.
  1453--1465, Sep 1996.

\bibitem{MAC-BA}
M.~Rezaeian and A.~Grant, ``Computation of total capacity for discrete
  memoryless multiple-access channels,'' \emph{IEEE Trans. Inf. Theory},
  vol.~50, no.~11, pp. 2779--2784, Nov 2004.

\bibitem{MAC-bound}
B.~Mamandipoor, K.~Moshksar, and A.~K. Khandani, ``Capacity-achieving
  distributions in gaussian multiple access channel with peak power
  constraints,'' \emph{IEEE Trans. Inf. Theory}, vol.~60, no.~10, pp.
  6080--6092, Oct 2014.

\bibitem{MAC-Com}
B.~Nazer and M.~Gastpar, ``Computation over multiple-access channels,''
  \emph{IEEE Trans. Inf. Theory}, vol.~53, no.~10, pp. 3498--3516, Oct 2007.

\bibitem{convex}
S.~Boyd and L.~Vandenberghe, \emph{{Convex optimization}}.\hskip 1em plus 0.5em
  minus 0.4em\relax Cambridge university press, 2004.

\bibitem{trust_region}
R.~H. Byrd, J.~C. Gilbert, and J.~Nocedal, ``{A trust region method based on
  interior point techniques for nonlinear programming},'' \emph{Mathematical
  Programming}, pp. 149--185, 2000.

\bibitem{MLC}
Z.~Chen, B.~Xia, Z.~Hu, and H.~Liu, ``Design and analysis of multi-level
  physical-layer network coding for gaussian two-way relay channels,''
  \emph{IEEE Trans. Commun.}, vol.~62, no.~6, pp. 1803--1817, June 2014.

\end{thebibliography}
\end{document}